\documentclass[twocolumn,aps,pra,preprintnumbers, mathtools, superscriptaddress, amsmath,amssymb]{revtex4}
\usepackage{graphicx}
\usepackage{amsthm, xcolor}

\newcommand*{\Tr}{{\rm Tr}\,}
\newcommand*{\diag}{{\rm diag}\,}
\renewcommand{\Xi}{H}

\newcommand{\ep}{u}

\newtheorem{theorem}{Theorem}

\begin{document}

\author{Denys I. Bondar}
\email{dbondar@tulane.edu}
\homepage{https://bondar.tulane.edu/}
\affiliation{Tulane University, New Orleans, LA 70118, USA}

\author{Alexander N. Pechen}
\email{apechen@gmail.com}
\homepage{www.mathnet.ru/eng/person17991}
\affiliation{Department of Mathematical Methods for Quantum Technologies, Steklov Mathematical Institute of Russian Academy of Sciences, Moscow 119991, Russia}
\affiliation{National University of Science and Technology "MISIS", Moscow 119049, Russia}

\title{Uncomputability and complexity of quantum control}

\date{\today}

\begin{abstract}
In laboratory and numerical experiments, physical quantities are known with a finite precision and described by rational numbers. Based on this, we deduce that quantum control problems both for open and closed systems are in general not algorithmically solvable, i.e., there is no algorithm that can decide whether dynamics of an arbitrary quantum system can be manipulated by accessible external interactions (coherent or dissipative) such that a chosen target reaches a desired value. This conclusion holds even for the relaxed requirement of the target only approximately attaining the desired value. These findings do not preclude an algorithmic solvability for a particular class of quantum control problems. Moreover, any quantum control problem can be made algorithmically solvable if the set of accessible interactions (i.e., controls) is rich enough. To arrive at these results, we develop a technique based on establishing the equivalence between quantum control problems and Diophantine equations, which are polynomial equations with integer coefficients and integer unknowns. In addition to proving uncomputability, this technique allows to construct quantum control problems belonging to different complexity classes. In particular, an example of the control problem involving a two-mode coherent field is shown to be NP-hard, contradicting a widely held believe that two-body problems are easy. 
\end{abstract}

\maketitle

\section*{\large Introduction}

Quantum control aims to find external actions (i.e., control policies) driving the dynamics of a quantum system such that a chosen target reaches a certain value, typically an extrema.
Consider either an open or closed quantum system with the density matrix $\hat{\rho}_t(u)$ at time $t$ evolving under the action of some time-dependent control $u = u(t)$. The following two control tasks play a prominent role: i) \emph{The problem of maximizing the expectation value of an observable} $\hat{O}$ at time $T$ is to find $u$ such that $\Tr[\hat{\rho}_T(u)\hat{O}]\to\max$. ii) \emph{The problem of a target density matrix preparation} $\hat{\rho}_{\rm f}$ is to construct $u$ such that $\| \hat{\rho}_T(u)-\hat{\rho}_{\rm f}\|^2\to\min$. Quantum control is of high interest due to fundamental aspects and many existing and prospective applications in quantum technologies including metrology, information processing, and matter manipulation~\cite{glaser2015training,butkovskiy1990control,shapiro2003principles,tannor2007introduction,fradkov2007cybernetical,d2007introduction,brif2010control,wiseman2009quantum,dong2010quantum}. Controls can be  continuous pulses~\cite{shapiro2003principles,brif2010control}, discrete transformations~\cite{bouten2009discrete,dong2011controllability,ticozzi2018alternating}, or combined continuous and discrete transformations, e.g., with instant quantum measurements~\cite{belavkin1999measurement,pechen2006quantum,shuang2007observation}.

The theory of Diophantine equations appears to be totally unrelated to quantum control. A Diophantine equation, $D(x_1, \ldots, x_n)=0$, is a polynomial equation with integer  coefficients solved with respect to positive integer unknowns $x_1, \ldots, x_n$. We note that this is a much more generic formulation than it originally appears. For example, an exponential Diophantine equation is a Diophantine equation additionally containing at least one term of the form $m^n$, where $m$ and $n$ are nonnegative integers and either both $n$, $m$ or $n$ alone are unknown.  Matiyasevich \cite{matiyasevich1993hilbert, matiyasevich2011can} has surprisingly shown that solving an exponential Diophantine equation is reducible to solving the (polynomial) Diophantine equation. Finding a rational solution of a polynomial equation with rational coefficients is also reducible to solving a Diophantine equation.

Diophantine equations are among the oldest branches of mathematics still actively studied. They also appear in quantum mechanics in a variety of contexts. For example, when deciding whether a quantum transition can be excited by a laser field consisting of $n$ commensurate frequencies~\cite{bondar2009photoelectron}. Utilizing the solutions of the exponential Diophantine Ramanujan-Nagell equation, Pavlyukh and Rau \cite{pavlyukh20131} established that only in the case of one and two qubit systems unitary transformations can be visualized as rotations. Kieu (Sec. 4 of \cite{kieu2003quantum}) noted that a Diophantine equation has a solution if and only if the Hamiltonian $\hat{H} = [D( \hat{a}^{\dagger}_1 \hat{a}_1, \ldots, \hat{a}^{\dagger}_n \hat{a}_n)]^2$ for $n$ noninteracting bosons has the zero ground state. Here $\hat{a}_j$ and $\hat{a}_j^{\dagger}$ are the creation and annihilation operators, respectively, for $j$-th boson. 
\begin{figure}[t]
\includegraphics[width=0.5\textwidth]{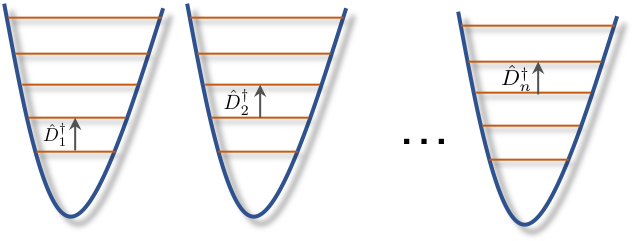}
\caption{A physical system for simulating Diophantine equations with $n$ variables. The system is either $n$ trapped ions or an $n$--mode coherent field. The controls $\hat{D}_1^{\dagger}$, $\ldots$, $\hat{D}_n^{\dagger}$ independently address each subsystem. For ions, the controls excite transitions between nearest levels, and transfer population of the highest excited state down to the ground state. For coherent states, the control for the $i$-th mode is the displacement $\hat D_i$ by the magnitude one. The Diophantine polynomial is embedded in the observable $\hat O$ whose expectation value has to be optimized as the control goal. A highly non-trivial example corresponds already to the simple case of a two-mode coherent field ($n=2$) with $\hat{O} = -(\alpha \hat{a}_1^{\dagger}\hat{a}_1^{\dagger}+ \beta \hat{a}_2^{\dagger} - \gamma ) (\alpha \hat{a}_1\hat{a}_1+ \beta \hat{a}_2 - \gamma )$, where $\alpha,\beta$, and $\gamma$ are positive integers. The observable is non-linear but physical; its leading term is of the Kerr-type nonlinearity. Maximizing the expectation of this observable is NP-hard, i.e., it is at least as hard as the famous Traveling Salesman Problem. Note that an $n=9$ system is sufficient to solve any Diophantine equation.}\label{Fig1}
\end{figure}

Diophantine equations are closely related to the theory of computability. A problem is called computable or decidable if in principle there exists an algorithm solving it. For the analysis of computability, the Turing machine~\cite{Turing1936} is a convenient mathematical model of the intuitive notion of algorithm; thus, the ``algorithm'' and ``computer program'' are used as the synonyms to  the Turing machine throughout. The link between Diophantine equations and computability is established by the Matiyasevich-Robinson-Davis-Putnam theorem, which gaves the negative answer to Hilbert's tenth problem \cite{matiyasevich1993hilbert, matiyasevich2011can}, meaning that there is no algorithm deciding whether an arbitrary given Diophantine equation is solvable. Furthermore, many open mathematical problems, including the Riemann hypothesis specifying the zeros of the Riemann zeta function, can be reformulated as questions about solvability of specially constructed Diophantine equations \cite{matiyasevichriemann}. It is noteworthy that the Riemann zeta function emerges in quantum statistical mechanics~\cite{bost1995hecke, planat2011riemann}, quantum entanglement and coherence~\cite{mack2010riemann, feiler2013entanglement, torosov2013quantum, feiler2015dirichlet}, random matrix theory~\cite{berry1986riemann, berry1999riemann}, string theory and related settings~\cite{lapidus2008search}. This enables a physical assessment of the Riemann hypothesis. Unfortunately, the required physical systems are not available off-the-shelf and need to be finessed, which remains a challenge. Recently a vigorous debate has been initiated by the proposal~\cite{bender2017hamiltonian} to reduce the Riemann hypothesis to the quantization of the classical Hamilton $2xp$.

Physics is also full of noncomputable problems. The undecidability of the presence of chaos in classical Hamiltonian systems has been established in \cite{da1991undecidability}. The problem whether a boolean combination of subspaces (including negations) is reachable by a quantum automation was proved to be undecidable~\cite{li2014decidable}. The question whether a quantum system is gapless also cannot be decided by an algorithm \cite{lloyd1993quantum,cubitt2015undecidability,bausch2018undecidability}. Whether a many-body model is frustration-free is undecidable as well \cite{Cubitt2011}. Smith (Sec. 6 of \cite{smith2006three}) identified a striking physical consequence of the Hilbert's tenth problem that ground state energies and half-life times of excited states are, strictly speaking, non-computable for many-body systems. A variety of seemingly simple problems in quantum information theory has been shown not to be decidable~\cite{wolf2011problems}. The question whether a sequence of outcomes of some sequential measurement cannot be observed is undecidable in quantum mechanics, whereas it is decidable in classical physics \cite{eisert2012quantum}. In this case, the algorithmic undecidability turned out to be the signature of quantumness.

Despite a significant interest to computability of various physical problems, to the best of our knowledge, computability of quantum control has not been systematically studied. The aim of this work is to fill this gap. We establish a connection between optimal quantum control and Diophantine equations and show how the latter emerges in control of various physical systems such as, e.g., a multi-mode coherent field driven by displacement operators of the fixed magnitude (Fig.~\ref{Fig1}). It is noteworthy that Diophantine equations were mentioned in~\cite{wolf2011problems} as a possible tool to analyze computability of quantum information tasks, but it has never been put to use. Furthermore, Theorem 2 in~\cite{wolf2011problems} may be interpreted to imply the undecidability of control task (i) for a finite dimensional quantum system even though ``quantum control'' is never mentioned in the preprint~\cite{wolf2011problems}. In this work, we not only prove uncomputability of quantum control by a different method and generalize to infinite dimensional quantum systems, but also show that even approximate quantum control tasks (i.e., for which it is sufficient to find an optimum with a  $\varepsilon$-accuracy) remain uncomputable as long as $\varepsilon$ is small enough. A general scheme to construct control tasks of various complexity classes is also developed.

In particular, we show that solving a Diophantine equation is equivalent to solving a certain quantum control task, and moreover, any question for which a computer program can give an answer can be stated as a quantum control task. This means that quantum control is Turing complete. In our approach the Diophantine equation is embedded in the target observable $\hat O$ whose expectation value has to be optimized as the control goal. This implies uncomputability of quantum control tasks (i) and (ii) introduced at the beginning. From a pragmatic point of view, this results means that there is no algorithm that outputs ``true'' or ``false'' whether a control sequence composed from a finite set of available controls exists to maximize either the observable's expectation or state-to-state transfer in an arbitrary \emph{generic} case. This, however, does not exclude the possibility that some \emph{particular} classes of control problems can have such an algorithm. The uncomputability motivates use of heuristics, e.g., such as machine learning \cite{palittapongarnpim2017learning}. Moreover, any control problem can be made algorithmically solvable if we deploy a sufficiently flexible controls.

Our technique based on establishing the equivalence between quantum control problems and Diophantine equations also enables knowledge transfer from the complexity theory for Diophantine equations to quantum control theory. In particular, one can construct control problems belonging to various complexity classes. A highly non-trivial example corresponds already to a seemingly simple case of two-mode coherent field ($n=2$) with target observable $\hat{O} = -(\alpha \hat{a}_1^{\dagger}\hat{a}_1^{\dagger}+ \beta \hat{a}_2^{\dagger} - \gamma ) (\alpha \hat{a}_1\hat{a}_1+ \beta \hat{a}_2 - \gamma )$, where $\alpha,\beta$, and $\gamma$ are positive integers. The controlled evolution is represented by a family of simple bosonic Gaussian channels. (Note that bosonic Gaussian channels play an important role in quantum information science~\cite{mari2014quantum}.) Maximizing the expectation of $\hat{O}$ is NP-hard, i.e., it is at least as hard as the famous Traveling Salesman Problem.

The rest of the paper is organized as follows: We proceed by giving a precise mathematical formulation of the quantum control problem. Then, we show how for a given quantum control problem to construct a Diophantine equation whose solution yields the optimal control policy. After that we demonstrate the converse: how to simulate a given Diophantine equation using quantum control. Finally, the uncomputability and complexity of the considered quantum control tasks are discussed.

\section*{\large Results}
\subsection*{Coherent Quantum Control} 

There are two physically distinct types of control regimes:  Coherent control exploits conservative forces, predominately, coherent electromagnetic fields (e.g., MRI and laser pulses)~\cite{shapiro2003principles,tannor2007introduction,dong2010quantum,brif2010control,soltamov2019excitation,niu2019universal,impens2019fast}, whereas quantum reservoir engineering utilizes nonconservative interactions with a thermostat~\cite{pechen2006teaching, gordon2008optimal, verstraete2009quantum, schirmer2010stabilizing, wiseman2011quantum, schmidt2011optimal,pechen2015measurement, koch2016controlling, vuglar2018nonconservative}. These different physical implementations have different mathematical formulations. Coherent control seeks a smooth temporal profile of the electromagnetic field steering dynamics, whereas reservoir engineering tailors a coupling between a thermostat and a controlled system.

Consider coherent control of an $n$-dimensional closed quantum system whose dynamics is governed by the Schr\"odinger equation for the unitary evolution operator $\hat{U}^\ep_t$
\begin{align}
	i\frac{d\hat{U}^\ep_t}{dt}=(\hat{H}_0+u(t)\hat{V})\hat{U}^\ep_t.
\end{align}
Here $\hat{H}_0$ and $\hat{V}$ are the free and interaction Hamiltonians and $\ep(t)$ is a time-dependent control (e.g., a shaped laser pulse). Interaction with several controls can be described similarly. The control $\ep(t)$ is commonly assumed to belong to some functional space, e.g., to the space $L^2([0,\infty])$ of square integrable functions of time. A key notion is the controllability of the system. A (closed) quantum system is pure state controllable if any two pure initial $|\psi_{\rm i}\rangle$ and final $|\psi_{\rm f}\rangle $ states can be connected by some control such that 
$|\psi_{\rm f}\rangle=\hat{U}_t^\ep|\psi_{\rm i}\rangle$ for some $t$ and $\ep$. 
The famous result is that the system is projective state 
controllable if and only if the Lie algebra 
${\rm Lie}\left\{ -i \hat{H}_0, -i \hat{V} \right\}$,
generated by all commutators of operators $-i\hat{H}_0$ and $-i\hat{V}$, is 
isomorphic to the Lie algebra $\mathfrak{sp}(n/2)$ 
or $\mathfrak{su}(n)$ for even $n$, or $\mathfrak{su}(n)$ for odd $n$.

A typical quantum control goal is to maximize objective $J_{\hat O}(\ep)=\langle \hat{O} \rangle$, which is expectation value of an observable $\hat O$. Without the loss of generality, the maximum value of $\langle \hat{O} \rangle$ (i.e., the largest eigenvalue of $\hat{O}$) can be assumed to be zero since adding a constant to $\hat O$ has no physical consequences. If the system is controllable, then this maximum value is attained by some control. Thus, the problem of existence of a globally optimal control $\ep^*$, i.e., such that $J_{\hat O}(\ep^*)=0$, is reduced to calculation of the rank of the Lie algebra ${\rm Lie}\left\{ -i \hat{H}_0, -i \hat{V} \right\}$, which can be done algorithmically for all finite-dimensional quantum controlled systems. In particular, this implies that \emph{establishing controllability of a quantum system is a computable task.}

\subsection*{Digitized Quantum Control} 
The assumption that controls belong to some infinite-dimensional functional space, while convenient mathematically, is not realistic from an engineering perspective. In experiments, one always has a limited finite number of available controls. A typical example is bang-bang control where $u(t)$ admits only two (on/off) values, or a switching control that can take several  values. Any continuous control function in experiments is approximated by some, usually small enough, set of values. Thus in laboratory and numerical experiments, both coherent and incoherent types of controls are digitized, which imposes the discretization and boundedness for the accessible values. The number of available controls  $N$ is always finite, albeit large. Moreover, the measured or computed values of physical quantities have a finite precision, and thus can be represented as rational numbers. The importance of this fact lead to the development of the $p$-adic mathematical physics~\cite{Volovich2010, Dragovich2017}. 

The most general state of a controlled quantum system is represented by a density matrix $\hat{\rho}$, which is a positive trace one operator in the system Hilbert space $\cal H$. The transformation of the system's initial density matrix $\hat{\rho}_0$ into the final density matrix under the action of the $i$-th control ($i=1,\ldots,N$) most generally can be represented by a Kraus map $\Phi_i$, i.e., a completely positive trace preserving transformation~\cite{kraus1983states}. Such maps have a (non-unique) operator-sum representation~\cite{nielsen2002quantum}
\begin{align}\label{Eq:KrausRepN}
	\Phi_i [ \hat\rho_0 ] &= \sum_{j} \hat{K}_{i,j} \hat\rho_0 \hat{K}_{i,j} ^{\dagger}.
\end{align}
Here $\hat K_{i,j}$ are (in general non-commuting) operators in $\cal H$ that satisfy the condition
$
\sum_{j} \hat{K}_{i,j}^{\dagger} \hat{K}_{i,j} = \mathbb{I}
$
to guarantee the trace preservation for the density matrix.

We define the \emph{Digitized Quantum Control} as a task of finding the control policy ${\bf p}$ specified by an integer sequence of length $P$, ${\bf p} = (p_1, p_2, \ldots, p_P) \in AP$, which is from some set of accessible policies $AP$, such that the propagated quantum state $\hat{\rho}({\bf p}) = \Phi_{p_P}\cdots \Phi_{p_1} [ \hat{\rho}_0]$ yields an extremum for a desired objective function. 

In particular, for the problem of preparing a target density matrix, when a quantum system is steered to a desired state $\hat{\rho}_{\rm f}$, one seeks the control policy ${\bf p}$ (if it exists) vanishing the objective function 
\begin{align}\label{Eq:OptProblemRho}
	F_{\hat\rho_{\rm f}}({\bf p}) = \| \hat{\rho}({\bf p}) -\hat{\rho}_{\rm f}\|^2.
\end{align}
This objective equals to zero if and only if $\hat{\rho}({\bf p})=\hat{\rho}_{\rm f}$.
A special yet equally important instance of the quantum state preparation problem is the problem of maximizing the expectation value of an observable $ \hat{O}$. Without the loss of generality, the maximum value of $\langle \hat{O} \rangle$ is set to zero. In this case, the goal is to find a control policy vanishing the objective function
\begin{align}\label{Eq:OptProblemFormulation}
	J_{\hat{O}}({\bf p})=\Tr\left( \hat{O} \hat{\rho}({\bf p}) \right).
\end{align}
A particular case is the problem of steering an initial pure state $|\psi_{\rm i}\rangle$ into a final pure state $|\psi_{\rm f}\rangle$. For this problem, $\hat O=|\psi_{\rm f}\rangle\langle\psi_{\rm f}|$ is the projector onto the final state and $\hat\rho_0=|\psi_{\rm i}\rangle\langle\psi_{\rm i}|$ is the projector on the initial state, and the objective becomes the transition probability, $J_{|\psi_{\rm f}\rangle\langle\psi_{\rm f}|}=P_{\rm i\to\rm f}=\langle\psi_{\rm f}|\hat\rho({\bf p})|\psi_{\rm f}\rangle$.

The functions \eqref{Eq:OptProblemRho} and \eqref{Eq:OptProblemFormulation} are related by the equality $F_{\hat\rho_{\rm f}}({\bf p})=\|\hat\rho_{\rm f}\|^2+\|\hat\rho({\bf p})\|^2-2J_{\hat\rho_{\rm f}}({\bf p})$, where $\|\hat\rho_{\rm f}\|^2$ is a constant independent of the control policy. In the general case, the problem of minimizing $F$ cannot be reduced to maximizing $J$, as illustrated by the example of a qubit with $\hat\rho_{\rm f}=\mathbb I/2$ for which  $J_{\hat\rho_{\rm f}}({\bf p})=1/2$ {is constant for any $\bf p$, while $F_{\hat\rho_{\rm f}}({\bf p})={\rm const}+\|\hat\rho({\bf p})\|^2$ is non-constant. However, in the case of pure initial $| \psi_{\rm i} \rangle$ and final  $| \psi_{\rm f}\rangle$ states  and controls restricted to unitary transformations, the problem \eqref{Eq:OptProblemRho} reduces to \eqref{Eq:OptProblemFormulation} with $\hat{O} = | \psi_{\rm f} \rangle\langle \psi_{\rm f}|-1$. Indeed, according to the Cauchy-Schwarz inequality, for an arbitrary state $| \phi \rangle$, $\langle \hat{O} \rangle = \langle \phi | \hat{O} | \phi \rangle \leq 0$ and the equality $\langle \hat{O} \rangle = 0$ takes place if and only if $| \phi \rangle = | \psi_{\rm f} \rangle$. The latter guarantees that the desired final state $| \psi_{\rm f}\rangle$ is reached once maximization of \eqref{Eq:OptProblemFormulation} is converged. Note that problem \eqref{Eq:OptProblemFormulation} is a special case of  problem \eqref{Eq:OptProblemRho} since minimizing \eqref{Eq:OptProblemFormulation} is equivalent to minimizing \eqref{Eq:OptProblemRho} with $\hat{\rho}_{\rm f}$ chosen as the projector onto the eigenstate corresponding to the largest eigenvalue of  $\hat O$.

Let us give examples of digitized quantum control. The simplest example is bang-bang control, which is switched on and off at some time instants, or more generally, is switched between two (maximal and minimal) values of its amplitude. In coherent control setting, steering quantum dynamics is achieved by tailoring the time profile of a laser pulse, whose intensity and bandwidth should not exceed engineering capabilities. The temporal form of the laser pulse can have the form $u(t)=\sum_{j=1}^P A_j\chi_{[t_j,t_{j+1}]}(t)$, where  $\chi_{[t_j,t_{j+1}]}$ is the characteristic function of the fixed time interval $[t_j,t_{j+1}]$ and $A_j$ is the pulse intensity at  the $j$-th time interval to be chosen among $N$ available pulse intensities. Another example is the field of the form $u(t)=\sum_{j=1}^P A_j\cos\omega_j t$, where $\omega_i$ are some fixed frequencies and the amplitudes $A_j$ are sought controls. In both cases the set  $AP$ of all attainable laser pulses has $N^P$ elements. In quantum computing, digitized quantum control mimics the problem of building a desired unitary transformation of a multi-qubit state using a finite number of universal quantum gates~\cite{lloyd1995almost,deutsch1995universality}, as well as various problems of finite computability~\cite{Jeandel2004universality}. The discrete-time control~\cite{belavkin1999measurement,pechen2006quantum,bouten2009discrete,dong2011controllability,ticozzi2018alternating} is also a particular type of digitized quantum control.

The digitized quantum control describes a very wide class of quantum control problems and has the following generic properties: 

(i)~The optimization problem \eqref{Eq:OptProblemFormulation}, in general, cannot be solved by the control policy of a finite length (see Theorem~\ref{Thm1} in Methods). This results follows from the fact that there is a continuum of digitized quantum control formulations, while finite-length controls form at most a countable set.

(ii)~For any observable and an arbitrary initial state, the relaxed condition $J_{\hat O} \approx 0$ can be satisfied with any desired error for a control policy of a finite length if the set of controls is rich enough (see~Theorem~\ref{Thm2} in Methods). For example, one can use the dissipative interaction to cool the quantum system to the ground state, and then rotate this state using a tailored unitary transformation constructed from a set of universal quantum gates to the state with $\langle \hat{O} \rangle \approx 0$ (i.e., to the eigenstate of $\hat O$ corresponding to the largest eigenvalue). 

According to the first property, time-discretization alone makes a quantum control problem ill-posed. For example, for closed systems there is no finite set of elementary unitary transformations that would exactly map a given pure initial state into a given pure target state. However, exact attainability of the target state is not required in practice. It is always sufficient to steer the system into a small neighborhood of the target state. The digitization, combining discretization and a finite precision, makes the problem well posed as per the second property. Theorem~\ref{Thm2} shows that for some (specially constructed) set of accessible controls such an approximate attainability is achievable in a finite number of steps. The finiteness, in turn, automatically implies algorithmic solvability since straightforward, albeit time consuming, looping through all the allowed control policies is guaranteed to terminate with uncovering the sought control strategy. The algorithmic solvability for this special control set, however, does not imply computability for any other set. In practical terms, Theorem~\ref{Thm2} means that any approximate quantum control task can be made solvable if the set of accessible controls is rich enough.

\subsection*{Reduction of digitized quantum control to a Diophantine equation}

As discussed above, in laboratory and numerical experiments elements of the matrices $\hat{K}_{i,k}$, $\hat{\rho}_0$, $\hat\rho_{\rm f}$ and $\hat{O}$ are complex numbers with rational imaginary and real parts. Using this fact, consider the matrix valued polynomials of the positive integer argument $i$
\begin{align}\label{Eq:LagrangeInterp}
	\hat{\phi}_j(i) = \sum_{l=1}^N \hat{K}_{l,j} \prod_{m=1, \, l \neq m}^N \frac{i-m}{l-m}.
\end{align}
By construction $\hat{\phi}_j(i) \equiv \hat{K}_{i,j}$ for $1 \leq i \leq N$. For every fixed $j$, equation \eqref{Eq:LagrangeInterp} is a matrix-valued Lagrange interpolation polynomial of variable $i$ passing trough $N$ points $(1,\hat K_{1,j}),\dots,(N,\hat K_{N,j})$. The objective function $F_{\hat\rho_{\rm f}}$ from Eq. \eqref{Eq:OptProblemRho} reduces to the following polynomial with rational coefficients of $P$ positive integer arguments
\begin{align}\label{Eq:DiophantineReductionF}
	\mathfrak{F}({\bf p}) = \left\| 
		\sum_{k_1, \dots, k_P} \left( \prod_{l=P}^1 \hat{\phi}_{k_l} (p_l) \right) \hat{\rho}_0 \left( \prod_{m=P}^1 \hat{\phi}_{k_m} (p_m) \right)^{\dagger} 
		- \hat{\rho}_{\rm f}
	\right\|^2 .
\end{align}
Finally, the policy ${\bf p}$ solves the state preparation problem  \eqref{Eq:OptProblemRho} if and only if it solves the Diophantine equation 
\begin{align}\label{Eq:DiophantineFEq}
	 \mathfrak{F}({\bf p})
	+ \prod_{{\bf p'} \in AC}  \sum_{k=1}^P \left( p_k - p'_k \right)^2 
	 = 0,
\end{align}
Both the terms in Eq. \eqref{Eq:DiophantineFEq} are non-negative. The sum equals to zero only if each term vanishes. The first term is zero if and only if the control policy ${\bf p}$ is optimal, i.e.,  $\hat\rho({\bf p})=\hat\rho_{\rm f}$. The last term equals to zero if and only if the control policy ${\bf p}$ is an accessible control, i.e., its components satisfy $p_k=p'_k$ for some ${\bf p'}\in AC$.

In the similar fashion, the policy ${\bf p}$ solves the problem of maximizing the expectation value \eqref{Eq:OptProblemFormulation} if and only if it solves the Diophantine equation 
\begin{align}\label{Eq:DiophantineQEq}
	 \mathfrak{J}^2({\bf p})
	+ \prod_{{\bf p'} \in AC}  \sum_{k=1}^P \left( p_k - p'_k \right)^2 
	 = 0,
\end{align}
where the reduction of the objective function $J_{\hat{O}}$ to a polynomial with rational coefficients reads
\begin{align}\label{Eq:DiophantineReduction1}
	\mathfrak{J}({\bf p}) = \sum_{k_1, \dots, k_P} \Tr\left( \hat{O} \left[ \prod_{l=P}^1 \hat{\phi}_{k_l} (p_l) \right] \hat{\rho}_0 \left[ \prod_{m=P}^1 \hat{\phi}_{k_m} (p_m) \right]^{\dagger} \right).
\end{align}
Note that there are many other ways to construct polynomials~(\ref{Eq:DiophantineReductionF}) and~(\ref{Eq:DiophantineReduction1}).

We now show that the problem of finding a control policy bringing the value of an objective function approximately close to the extremum (with a desired degree of accuracy) is also reducible to solving a Diophantine equation. To present the derivation in a unified fashion, let $\mathfrak{G}({\bf p})$ denote either $\mathfrak{J}^2({\bf p})$  or $\mathfrak{F}({\bf p})$. By definition $\mathfrak{G}({\bf p}) \geq 0$ [i.e., it follows from Eqs.~\eqref{Eq:DiophantineReduction1} and \eqref{Eq:DiophantineReductionF}]. Thus, reaching close to the optimum means that for a given  rational $\varepsilon > 0$, $\mathfrak{G}({\bf p}) < \varepsilon$. Using the methods presented in Ch.~1 of Ref. \cite{matiyasevich1993hilbert}, we observe that the inequality is satisfied if and only if there are \emph{positive} integers $a$ and $b$ such that $\mathfrak{G}({\bf p}) + a/b = \varepsilon$. Lagrange's four-square theorem, stating that a positive integer can be written as the sum of four integer squares, lifts the requirement of the positivity of $a$ and $b$, leading to the following Diophantine equation with ancillary unknowns $a_1, \ldots, a_4, b_1, \ldots, b_4$:
\begin{align}
	\left(b_1^2 + b_2^2 + b_3^2 + b_4^2 \right)\left(\mathfrak{G}({\bf p}) - \varepsilon\right)
	+ a_1^2 + a_2^2 + a_3^2 +a_4^2 = 0.
\end{align}
Finally, by constraining control policies to the accessible set $AP$, we obtain the sought Diophantine equation for the control policy ${\bf p}$ yielding the optimum within an $\varepsilon$-accuracy
\begin{align}
	&\left[
		\left(b_1^2 + b_2^2 + b_3^2 + b_4^2 \right)\left(\mathfrak{G}({\bf p}) - \varepsilon\right)
		+ a_1^2 + a_2^2 + a_3^2 +a_4^2 
	\right]^2 \notag\\
	& \qquad + \prod_{{\bf p'} \in AC}  \sum_{k=1}^P \left( p_k - p'_k \right)^2  = 0.
\end{align}

\subsection*{Simulation of a Diophantine equation with digitized quantum control}

Here we show how to simulate the problem of finding positive integer solutions of a Diophantine equation $D(x_1, \ldots, x_n) = 0$ with digitized quantum control. Let us introduce $X$-dimensional vectors $| e_k \rangle$ containing $1$ in the $k$-th position and zeros elsewhere, the matrix $\hat{\Xi} =\diag(1, 2, \ldots, X)$ and the unitary $X \times X$ shift matrix
\begin{align}
	\hat{\Sigma} = \begin{pmatrix}
			0	&   &	 &  1 \\
			1	& 0	&  & \\
				  & \ddots & \ddots & \\
					&	& 1	& 0
		\end{pmatrix}
\end{align}
obeying $\hat{\Sigma} | e_k \rangle =| e_{k+1} \rangle$, where $| e_{X+1} \rangle = | e_1 \rangle$ is assumed. Define also for $l=1,\ldots,n$
\begin{align}
	\left( \hat{\Xi}_l \atop \hat{\Sigma}_l \right) = \bigotimes_{q=1}^{l-1} \hat{1}
		\otimes \left( \hat{\Xi} \atop \hat{\Sigma} \right)  
		\bigotimes_{q=1}^{n-l} \hat{1}.
\end{align}
Since all matrices $\hat{\Xi}_l$ commute by construction, the eigenvalues and eigenvectors of the matrix $\hat{D} = D(\hat{\Xi}_1, \ldots, \hat{\Xi}_n)$ are given by
\begin{align}
	 \left( D(\hat{\Xi}_1, \ldots, \hat{\Xi}_n)
	-D(x_1, \ldots, x_n) \right) | e_{x_1}, \ldots, e_{x_n} \rangle = 0,
\end{align}
where $1 \leq x_l \leq X$. This relation allows to formulate the equivalence between Diophantine equations and quantum control:

A Diophantine equation $D(x_1, \ldots, x_n) = 0$ has a positive integer solution with $1 \leq x_l \leq X$ if and only if
the control problem \eqref{Eq:OptProblemFormulation} with $\hat{O} = -D(\hat{\Xi}_1, \ldots, \hat{\Xi}_n)^2$,
$\hat{\rho}_0 = \bigotimes_{l=1}^n | e_1 \rangle\langle e_1|$, $\Phi_0 [ \hat\rho ] = \hat\rho$, $\Phi_l [ \hat\rho ] = \hat{\Sigma}_l \hat\rho \hat{\Sigma}_l^{\dagger}$, $l=1,\ldots,n$ has a policy ${\bf p}$ yielding $J_{\hat O}=0$. The set of accessible policies is $AP = [0, 1, 2, \ldots, n]^Q$, where $Q \geq nX$. 

The motivation for this construction is as follows (see also Fig.~\ref{Fig1}): The vector $| e_k \rangle$ encodes integer $k$ as $\hat{\Xi} | e_k \rangle = k | e_k \rangle$; similarly, vector $| e_{x_1}, \ldots, e_{x_n} \rangle$ encodes an integer tuple $(x_1, \ldots, x_n)$. The initial density matrix $\hat{\rho}_0$ represents the $n$-tuple $(1, \ldots, 1)$. The action of each $\Phi_l$ onto a density matrix encoding a tuple $(x_1, \ldots, x_n)$ is equivalent to the operation $x_l \to x_l + 1$ of incrementing the tuple's $l$-th component. To scan all the values of $x_l$ from $1$ to $X$, $\Phi_l$ needs to sequentially act  $X$ times onto $\hat{\rho}_0$. Thus, the length of the policy should be at least $nX$ to scan through all possible combinations of the $n$ variables.  The trivial identity transformation $\Phi_0$ (not modifying the density matrix) is employed due to the following reason: Assume the value of the $l'$-th component of the solution of the Diophantine equation is $L < X$, then 
$\Phi_{l'}$ should be used only $L$ times followed by $(X-L)$ applications of $\Phi_0$. 

The construction above employs only unitary operations. However, the described method can be adopted to use the amplitude damping Kraus maps \cite{nielsen2002quantum}. Consider the Kraus map 
$\Phi [ \hat\rho ] = \sum_{i=1}^X \hat{K}_i \hat\rho  \hat{K}_i^{\dagger}$, where
$\hat{K}_1 = | e_1 \rangle\langle e_1 |$, 
$\hat{K}_i =  | e_{i-1} \rangle\langle e_i |$, $i=2,\ldots,X$ obeying $\Phi(  | e_k \rangle\langle e_k | ) = | e_{\max(1, k-1)} \rangle\langle e_{\max(1, k-1)} |$, which moves the population from $k$ to $(k-1)$ level with the first level being the stationary state. Then a Diophantine equation $D(x_1, \ldots, x_n) = 0$ has a positive integer solution, $1 \leq x_l \leq X$, if and only if
the optimization problem \eqref{Eq:OptProblemFormulation} with $\hat{O} = -D(\hat{\Xi}_1, \ldots, \hat{\Xi}_n)^2$,
$\hat{\rho}_0 = \bigotimes_{l=1}^n | e_X  \rangle\langle e_X |$, $\Phi_0 [ \hat\rho ] = \hat\rho$, $\Phi_i [ \hat\rho] = \sum_{j=1}^X \hat{K}_{i,j} \hat\rho\hat{K}_{i,j}^{\dagger}$,
$
\hat{K}_{l,k} = \bigotimes_{q=1}^{l-1} \hat{1}\otimes \hat{K}_k  \bigotimes_{q=1}^{n-l} \hat{1}
$, $l=1,\ldots,n$ has a policy ${\bf p}$ yielding $J_{\hat{O}}=0$. The set of accessible policies is $AP = [0, 1, 2, \ldots, n]^Q$, where $Q \geq nX$. 

The two presented constructions have a drawback that they involve the upper bound $X$ for a sought solution. The following third reduction of the Diophantine equation into digitized quantum control uses multi-mode coherent states and overcomes this limitation (see also Fig.~\ref{Fig1}): 
\begin{theorem}
A Diophantine equation $D(x_1, \ldots, x_n) = 0$ is solvable in nonnegative integers if the optimization problem \eqref{Eq:OptProblemFormulation} with $\hat{\rho}_0 = | 0, \ldots, 0\rangle\langle 0, \ldots, 0|$, $\Phi_l [ \hat\rho ] = \hat{D}_l^{} \hat\rho \hat{D}_l^{\dagger}$,  $\hat{D}_l = \exp(\hat{a}_l^{\dagger} - \hat{a}_l)$ and a control policy of an arbitrary length  achieves $J_{\hat{O}} = 0$ for the observable  $\hat{O} = -D(  \hat{a}_1, \ldots,  \hat{a}_n)^{\dagger} D(  \hat{a}_1, \ldots,  \hat{a}_n)$. 
\end{theorem}
The controlled Kraus map $\Phi_l$ is a bosonic Gaussian channel~\cite{mari2014quantum}. Here $| \alpha_1, \ldots, \alpha_n \rangle$ is a composite coherent state: $\hat{a}_l | \alpha_1, \ldots, \alpha_n \rangle = \alpha_l | \alpha_1, \ldots, \alpha_n \rangle$, so that $D(  \hat{a}_1, \ldots,  \hat{a}_n) | \alpha_1, \ldots, \alpha_n \rangle = D(\alpha_1, \ldots, \alpha_n) | \alpha_1, \ldots, \alpha_n \rangle$. The displacement operator $\hat{D}_l$ acts on $l$-th mode as $\hat{D}_l | \ldots, \alpha_l, \ldots \rangle = | \ldots, \alpha_l + 1, \ldots \rangle$ and describes the increase of the laser intensity by the magnitude one for the $l$-th mode without altering the phase. Thus, the maximum of the objective function, $J_{\hat O}=0$, is reached only by the coherent state $| x_1, \ldots, x_n \rangle$ such that $D(x_1, \ldots, x_n) = 0$. Unlike number states, which are difficult to create experimentally, this reduction uses only easily available coherent states describing laser radiation. The presented formulation is open to a number of generalizations. 

Solving a Diophantine equation can also be reduced to solving an optimization problem with an $\varepsilon$-accuracy for a sufficiently small $\varepsilon$. This readily follows from the above constructions. Chosen controls and initial quantum state ensure that $\langle \hat{O} \rangle$ takes only non-positive integer values during  optimization. Thus, for any $0 < \varepsilon < 1$, finding a control policy that $ -\varepsilon\le \langle \hat{O} \rangle\le 0$ (i.e., optimal with an $\varepsilon$-accuracy) is equivalent to reaching the maximum $\langle \hat{O} \rangle = 0$. This conclusion holds for $\varepsilon<1$ which is relatively small in comparison with typical values of the objective (i.e., values of the Diophantine polynomial evaluated at integer arguments). If the allowed error is large enough then a separate analysis is needed, which is beyond the scope of this work.

The found equivalence of Diophantine equations and digitized quantum control  employs commuting Kraus maps. They can be viewed as a faithful matrix formulation of the special Turing machine that was constructed in section 5.4 of \cite{matiyasevich1993hilbert}. This Turing machine for any given Diophantine equation loops through all the tuples of positive integers and halts when a solution is found. The formulation obtained in the present work is not necessarily optimal, e.g., in the number of used Kraus maps. Constructing more compact representations relying on non-commutative operators and using quantum interferences should be a subject of future work. Furthermore, the presented reduction transfers the complexity of a Diophantine equation into the observable while keeping controls simple. Different constructions that distribute the complexity between the control and observable should be investigated.

By including additional unknowns any Diophantine equation can be transformed to the equivalent forth order equation~\cite{matiyasevich1993hilbert}. The number of unknowns can be decreased down to 9 by increasing the degree of the Diophantine polynomial~\cite{jones1982universal}. These observations imply that for digitized quantum control with multimode coherent states, it is always sufficient to use no more than 9 modes by utilizing a higher order polynomial observable; whereas, the order of nonlinearities can be decreased to 4 by increasing the number of modes.

\subsection*{Uncomputability and complexity of digitized quantum control}

The Matiyasevich-Robinson-Davis-Putnam theorem~\cite{matiyasevich1993hilbert, matiyasevich2011can} uncovers an equivalence between sets of solutions of Diophantine equations and sets of outputs of computer programs, which are allowed to run forever. The reduction of a computer program to the corresponding Diophantine equation is constructive. The Matiyasevich-Robinson-Davis-Putnam theorem leads to the negative resolution of Hilbert's tenth problem \cite{matiyasevich1993hilbert, matiyasevich2011can}, meaning that the solvability of an arbitrary Diophantine equation is not decidable.

The negative resolution of Hilbert's tenth problem also implies the uncomputability of digitized quantum control. Moreover, digitized quantum control with an $\varepsilon$-accuracy for small enough $\varepsilon$ are also uncomputable.

Thus, the problem of maximizing an expectation value~\eqref{Eq:OptProblemFormulation} is undecidable and so is the problem of quantum state preparation \eqref{Eq:OptProblemRho} since the former is a special case of the latter. This means that there is no algorithm deciding on the existence or non-existence of an optimal control solution for an arbitrary digitized quantum control problem. This finding does not preclude an algorithmic solution for a particular digitized quantum control problem. It is noteworthy that tracking the time-evolution of an observable \cite{campos2017make, magann2018singularity} is a manifestly algorithmically solvable quantum control problem, which nevertheless cannot be reduced to either objective function \eqref{Eq:OptProblemRho} or \eqref{Eq:OptProblemFormulation}.

The established equivalence between digitized quantum control and Diophantine equations can be used to synthesize quantum control problems belonging to a certain computational complexity class. For example, finding an optimal control policy to reach $\langle \hat{O} \rangle = 0$ with $\hat{O} = -(\alpha \hat{a}_1^{\dagger}\hat{a}_1^{\dagger}+ \beta \hat{a}_2^{\dagger} - \gamma ) (\alpha \hat{a}_1\hat{a}_1+ \beta \hat{a}_2 - \gamma )$ is NP-hard. This is a consequence of the fact that it is an NP-complete problem to decide the solvability of the Diophantine equation $\alpha x_1^2 + \beta x_2 = \gamma$ with respect to $x_1$ and $x_2$ \cite{adleman1977taking}. Therefore, this digitized quantum control problem is at least as hard as the celebrated Traveling Salesman Problem. Note that the leading nonlinearity in $\hat{O}$ is of the Kerr type (see, e.g.,~\cite{sanders1992quantum,puri2017engineering}), which makes this proposal of experimental interest. 

On the contrary, our technique can also be used to construct simple quantum control problems. For example, finding an optimal control policy to reach $\langle \hat{O} \rangle = 0$ with $\hat{O} = -(\hat{a}_1^{\dagger}\hat{a}_1^{\dagger}-n^2 \hat{a}_2^{\dagger}\hat{a}_2^{\dagger} - 1) (\hat{a}_1\hat{a}_1-n^2 \hat{a}_2\hat{a}_2 - 1)$ is easy because $x_1=1, x_2=0$ is the only solution of the Pell's Diophantine equation $x_1^2 - n^2 x_2^2 = 1$ in non-negative integers. 

It is worth comparing our findings with the theory of quantum control landscapes~\cite{rabitz2004quantum}, which studies the objective, e.g., $J_{\hat O}(u)$, as a functional of the control $u=u(t)$, which is an arbitrary time dependent function not restricted to integer sequences. If the objective has only global maxima (i.e., local maxima are absent), then a gradient algorithm converges to an optimal control. In this work, we consider a different situation when there is a finite number of basic elementary controls that can be applied multiple times and in an arbitrary time order. This is the case of digitized quantum control, for which no algorithm can find a solution in the general case. In Ref. \cite{day2019glassy}, it has been shown how the control landscape of digitized quantum control problems can be mapped to a classical Ising model, which among other things reveals that the connection between easy and hard control tasks is akin to a phase transition.

Our reduction of Diophantine equations to quantum control problems allows to formulate various mathematical conjectures as, perhaps extremely complex, quantum control tasks, thereby in principle providing a route for experimental assessment for the unsolved conjectures. As an example, consider the Riemann hypothesis. Matiyasevich wrote the explicit form of the Diophantine equation with the property that it has infinitely many solutions if the hypothesis is false and has no solution if the hypothesis is true~\cite{matiyasevichriemann}. Using our approach the Riemann hypothesis can be reformulates as some  quantum control problem albeit extremely complex. Moreover, quantum control can be applied to evaluate any mathematical expression formed from arithmetic ($+$, $\times$, $-$, $=$) and logical ($>$, ``and'', ``or'') operations, existential quantifiers (e.g., $\exists x$ -- there exists $x$), and bounded universal quantifiers (e.g., $\forall x < M$ -- for all $x$ less than $M$). A constructive proof that such expressions are equivalent to solving Diophantine equations can be found in chapters 1 and 6 of \cite{matiyasevich1993hilbert}. 

\section*{\large Discussion}

Computability of quantum control problems has been analyzed. A realistic situation, when a number of controls is finite, has been considered. We have shown that within this setting solving quantum control problems is equivalent to solving Diophantine equations. As a consequence, quantum control is Turing complete. The established equivalence is a new technique for quantum technology that, e.g., allows to construct quantum problems belonging to a specific complexity class. Examples of a multimode coherent field control are explicitly constructed. The negative answer to the Hilbert's tenth problem implies that there is no algorithm deciding whether there is a control policy connecting two quantum states represented by arbitrary pure or mixed density matrices, i.e., the most general fixed-time quantum state-to-state control problem is not algorithmically solvable. This result applies to the problems of finding exact as well approximate solutions for sufficiently small errors. Our method opens up an opportunity to recast many open mathematical problems, including the Riemann hypothesis, as quantum control tasks. The uncovered non-algorithmic nature makes quantum control a fruitful research area. 

\section*{\large Methods}

We shall prove two theorems elucidating properties of digitized quantum control.
\begin{theorem}\label{Thm1}
There are digitized quantum control problems for which the condition $J_{\hat O}({\bf p})=J_0$ for some $J_0$ such that $O_{\rm min}\le J_0\le O_{\rm max}$ ($O_{\rm min}$ and $O_{\rm max}$ are minimal and maximal eigenvalues of $\hat O$) is never satisfied by any control policy of finite length. 
\end{theorem}
\begin{proof}
Consider a specific class of digitized quantum control problems with Kraus maps $\Phi_i$ being unitary rotations, $\Phi_i[\hat\rho] =  \hat{V}_i \hat\rho\hat{V}_i^{\dagger}$, $ \hat{V}_i \hat{V}_i^{\dagger} = \hat{\mathbb I}$, and the initial density matrix $\hat{\rho}_0 = | \psi \rangle\langle \psi |$ corresponding to a pure state. In this case, the negation of the theorem statement implies that for every $\hat{O}$ and $| \psi \rangle$ there exists a policy ${\bf p}$ of a finite length $P$ such that $J_{\hat O}({\bf p})=J_0$, where
\begin{align}
	J_{\hat O} ({\bf p})= \langle \psi | \left( \prod_{k=1}^P\hat{V}_{p_k} \right)^{\dagger} \hat{O}  \left( \prod_{k=1}^P\hat{V}_{p_k} \right) | \psi \rangle. 
\end{align}
Comparing this with the eigendecomposition of the observable $\hat{O} = \hat{U}^{\dagger} \diag(O_{\min}, \ldots, O_{\max}) \hat{U}$, $\hat{U}\hat{U}^{\dagger} =\hat{\mathbb I}$, we conclude that the condition $J_{\hat O}=J_0$ can be met if we select $	\hat{U} =\prod_{k=1}^P\hat{V}_{p_k}$ and 
$
	| \psi \rangle = \sqrt{\lambda} | \psi_{\min} \rangle + \sqrt{1-\lambda} | \psi_{\max} \rangle,
$
where $| \psi_{\max} \rangle$ and $| \psi_{\min} \rangle$ are the normalized eigenvectors corresponding to the largest and smallest eigenvalues of $\hat{O}$ and $\lambda = (O_{\max} - J_0) / (O_{\max} - O_{\min})$.

The latter establishes a correspondence between an arbitrary unitary matrix and a finite integer sequence ${\bf p}$. Thus we reached the contradiction that the set of all unitary matrices is countable.
\end{proof}

\begin{theorem}\label{Thm2}
For an $n$-dimensional quantum system, there exists a finite set of Kraus map controls such that for any $\hat{\rho}_0$, $\hat{O}$, $J_0$ ($O_{\rm min}\le J_0\le O_{\rm max}$), and an arbitrary $\epsilon>0$, there is a control policy of a finite length satisfying $|J_{\hat O}-J_0|<\epsilon$.
\end{theorem}
\begin{proof}
Let $|\tilde\psi\rangle$ be any pure state of the quantum system. Then there exists a universally optimal Kraus map~\cite{wu2007controllability} $\Phi_{\tilde\psi}$ such that $\Phi_{\tilde\psi}(\hat\rho)=|\tilde\psi\rangle\langle\tilde\psi|$ for any density matrix $\hat{\rho}$. Kraus operators for this universally optimal Kraus map have the form $\hat{K}_i=|\tilde\psi\rangle\langle\chi_i|$, where $\{ |\chi_i\rangle \}_{i=1}^n$ is an orthonormal basis in the system Hilbert space.

By the Solovay-Kitaev theorem~\cite{kitaev1997quantum, kitaev2002classical, nielsen2002quantum}, for an $n$-dimensional quantum system there exists a finite set $\cal U$ of unitary operators such that for any unitary operator $\hat{U} \in SU(n)$ there exists a finite sequence $\hat{S}=\hat{S}_k\cdots \hat{S}_1$ with elements $\hat{S}_i\in\cal U$ that satisfies $d(\hat{U}, \hat{S})\equiv\| \hat{U} - \hat{S} \|\equiv\sup\limits_{\|\psi\|=1}\|(\hat{U}-\hat{S})\psi\|<\epsilon$.  Let the corresponding set of Kraus maps be ${\cal K}=\{\Phi|\Phi(\hat\rho)=\hat{U} \hat\rho \hat{U}^\dagger, \textrm{ where } \hat{U} \in{\cal U}\}$.

Consider the (finite) set $\tilde{\cal K}=\{\Phi_{\tilde\psi}\}\cup {\cal K}$ of Kraus maps, where one map is non-unitary and all other are unitary. Then the constructed set satisfies the statement of the theorem. Indeed, let $|\psi\rangle$ be the vector constructed in the proof of Theorem~\ref{Thm1}. Let $\hat{U}$ be a unitary operator such that $\hat{U}|\tilde\psi\rangle=|\psi\rangle$, and $\hat{S}=\hat{S}_k\cdots \hat{S}_1$ be its $\epsilon/\|\hat O\|$-approximation by elements of $\cal U$. Then for the finite composition $\Phi=\Phi_k\dots\Phi_1\Phi_{\tilde\psi}$ we have $|J_{\hat O}-J_0|<\epsilon$.

We remark that, by construction, there are infinitely many sets $\tilde{\cal K}$ satisfying the theorem. Indeed, the map $\Phi_{\tilde\psi}$ can be chosen for any vector $|\tilde\psi\rangle$ and moreover, there exist infinitely many sets of unitary operators $\cal U$. 
\end{proof}

\section*{\large Acknowledgments}
D.I.B. is grateful to Prof. Wolfgang Schleich for igniting an interest in the overlap between quantum mechanics and number theory, which has directly lead to the current work. The results for uncomputability and complexity of controlling open quantum systems are obtained with the support of the Russian Science Foundation project 17-11-01388 at Steklov Mathematical Institute. The rest is supported by the Humboldt Research Fellowship for Experienced Researchers, the Army Research Office (ARO) (grant W911NF-19-1-0377), and Defense Advanced Research Projects Agency (DARPA) (grant D19AP00043) for D.I.B. and by project 1.669.2016/1.4 of the Ministry of Science and Higher Education of the Russian Federation for A.P. The views and conclusions contained in this document are those of the authors and should not be interpreted as representing the official policies, either expressed or implied, of ARO, DARPA, or the U.S. Government. The U.S. Government is authorized to reproduce and distribute reprints for Government  purposes notwithstanding any copyright  notation herein. 

\bibliography{DiophantineQuantum}

\end{document}